\documentclass[11pt]{amsart}
\usepackage{graphicx,amssymb,amsmath,amsthm}
\usepackage{enumerate}
\usepackage{dsfont}
\usepackage[colorlinks, citecolor=red]{hyperref}
\usepackage{rotating}
\usepackage{mathrsfs}
\usepackage{epsfig}
\usepackage{lscape}
\usepackage{subfigure}
\usepackage{epstopdf}
\usepackage{multirow}

\usepackage{caption}
\usepackage{algorithm}
\usepackage{algpseudocode}

\usepackage{geometry}
\usepackage{amsfonts}
\usepackage{graphicx,cite}

\usepackage{longtable}
\usepackage{paralist}
\usepackage{enumerate}
\usepackage{url}

\newcommand{\abs}[1]{\lvert#1\rvert}

\newcommand{\argmin}[1]{\mathop{\rm argmin}\limits_{#1}}
\newcommand{\argmax}[1]{\mathop{\rm argmax}\limits_{#1}}

\newcommand{\ve}{{\mathbf e}}

\newcommand{\R}{{\mathbb R}}

\newcommand{\cS}{{\mathcal S}}

\newcommand{\rank}{{\rm rank}}

\newtheorem{definition}{Definition}[section]
\newtheorem{corollary}[definition]{Corollary}

\newtheorem{theorem}[definition]{Theorem}
\newtheorem{lemma}[definition]{Lemma}

\newtheorem{problem}[definition]{Problem}
\date{}

\begin{document}
\baselineskip 18pt
\bibliographystyle{plain}
\title[Subset selection for matrices with Fixed Blocks]
{Subset Selection for Matrices with Fixed Blocks}

\author{Jiaxin Xie}
\address{School of Mathematical Sciences, Beihang University, Beijing, 100191, China }
\email{xiejx@buaa.edu.cn}

\author{Zhiqiang Xu}
\thanks{Zhiqiang Xu was supported  by NSFC grant (91630203, 11688101),
Beijing Natural Science Foundation (Z180002).}
\address{LSEC, Inst.~Comp.~Math., Academy of
Mathematics and System Science,  Chinese Academy of Sciences, Beijing, 100091, China \newline
School of Mathematical Sciences, University of Chinese Academy of Sciences, Beijing 100049, China}
\email{xuzq@lsec.cc.ac.cn}

\begin{abstract}
Subset selection for  matrices is the task of extracting a column sub-matrix from a
given matrix $B\in\mathbb{R}^{n\times m}$ with $m>n$ such that the pseudoinverse of
the sampled matrix has as small Frobenius or spectral norm as possible. In this
paper, we consider  a more general problem of subset selection for matrices that
allows a block to be fixed  at the beginning. Under this setting, we provide a
deterministic method for  selecting  a column sub-matrix from $B$. We also present a
bound for both the Frobenius and  spectral  norms of the pseudoinverse of the sampled
matrix, showing that the bound is asymptotically optimal.
  The main technology  for proving  this result is the  interlacing  families   of polynomials
 developed  by Marcus, Spielman,  and Srivastava. This idea also results   in a
  deterministic greedy selection algorithm that produces the sub-matrix promised by our result.
\end{abstract}
\maketitle

\section{Introduction}
\subsection{Subset selection for matrices}
Subset  selection for  matrices aims to  select a column sub-matrix  from a given
matrix $B\in\mathbb{R}^{n\times m}$ with $m>n$ such that the sampled matrix is
well-conditioned. For convenience, we  assume that $B$ is \emph{full-rank}, i.e.,
$\mbox{rank}(B)=n$. Given $\cS \subseteq [m]:=\{1,\ldots,m\}$, the cardinality of the
set $\cS$ is denoted by $|\cS|$. We use $B_{\cS}$ to denote the sub-matrix of $B$
obtained by extracting the columns of $B$ indexed by $\cS$ and use
$B_{\mathcal{S}}^{\dagger}$ to denote the Moore-Penrose pseudoinverse of $B_{\cS}$.
We use $\|B\|_2:=\max\limits_{\|x\|_2=1} \|Bx\|_2$  and
$\|B\|_F:=\sqrt{\mbox{Tr}(BB^T)}$ to denote, respectively, the spectral norm  and
Frobenius norm of $B$. Let $k\in [n,m-1]:=\{n,\ldots,m-1\}$ be a sampling parameter.
We can formulate the subset selection for matrices as follows.

\begin{problem}
\label{SSM} Find a subset $\cS_{opt}\subset\{1,2,\ldots,m\}$ with  cardinality at
most $k$ such that $\mbox{rank}(B_{\mathcal{S}_{opt}})=\mbox{rank}(B)$ and
$\|B_{\mathcal{S}_{opt}}^{\dagger}\|_{\xi}^2$ is minimized, i.e.,
 $$
 \mathcal{S}_{opt}\in\argmin{\mathcal{S}\in\mathcal{F}(B,k)}\|B_{\mathcal{S}}^{\dagger}\|_{\xi}^2,
 $$
 where $\mathcal{F}(B,k)=\{\mathcal{S}:|\mathcal{S}|\leq k,
 \mbox{rank}(B_{\mathcal{S}})=\mbox{rank}(B)\}$ and
 $\xi=2$ and $\mbox{F}$ denotes the spectral and  Frobenius matrix norm, respectively.
\end{problem}
Problem \ref{SSM} is raised in many applied areas, such as preconditioning for
solving linear systems\cite{arioli2014},   sensor selection \cite{joshi2009}, graph
signal processing \cite{chenD2015,zhao2016}, and feature selection in $k$-means
clustering \cite{boutsidis2013,boutsidis2011}. In \cite{avron2013}, Avron and
Boutsidis showed an interesting  connection between Problem \ref{SSM}  and the
combinatorial problem of finding a low-stretch spanning tree in an undirected graph.
In the statistics literature, the subset selection  problem has also been studied.
For instance,  the solution to  Problem \ref{SSM} has statistically optimal design
for linear regression provided  $\xi=F$  \cite{Derezinski2018,pukelsheim2006}.

One simple method for solving Problem \ref{SSM} is to evaluate  the performance of
all ${m \choose k} $ possible subsets with size $k$, but evidently it is
computationally expensive unless $m$ or $k$ is very small. In \cite{civril2009}, \c
Civril and Magdon-Ismail studied the complexity of the spectral norm version of
Problem \ref{SSM}, where they showed that it is NP-hard. Several heuristics have been
proposed to approximately solve the subset selection problem (see Section
\ref{related-work}).

\subsection{Our contribution}

In this paper, we consider a \emph{generalized version} of  subset selection for
matrices, where
 we have a matrix $A$ fixed at first, and our goal is to supplement   this matrix  by adding columns of $B$ such that $[A \ B_{\mathcal{S}}]$ has as small Frobenius or spectral norm as possible. Usually, $A$ is chosen  as a column sub-matrix of $B$.
 This notion of keeping a fixed block of $B$ is useful
 if we already know that such a block has some distinguished properties.
We state the problem as follows:

\begin{problem}
\label{GSSM}
Suppose that $A\in\mathbb{R}^{n\times \ell}$ and $B\in\mathbb{R}^{n\times m}$ with $\mbox{rank}(A)=r$ and $rank\big([A \  B]\big)=n$.
Find a subset $\cS_{opt}\subset\{1,2,\ldots,m\}$ with cardinality at most $k\in [n-r,m-1]$ such that $\mbox{rank}([A\ B_{\mathcal{S}_{opt}}])=\mbox{rank}([A\ B])$ and $\|[A\ B_{\mathcal{S}_{opt}}]^{\dagger}\|_{\xi}^2$ is minimized, i.e.,
 $$
 \mathcal{S}_{opt}\in\argmin{\mathcal{S}\in\mathcal{F}(B,k)}\big\|([A\ B_{\mathcal{S}}])^{\dagger}\big\|_{\xi}^2,
 $$
 where  $\mathcal{F}(B,k)=\{\mathcal{S}:|\mathcal{S}|\leq k, \mbox{rank}([A\ B_{\mathcal{S}}])=\mbox{rank}([A\ B])\}$ and $\xi=2$ and $\mbox{F}$ denotes the spectral and  Frobenius matrix norm, respectively.
\end{problem}

We would like to mention that the Frobenius norm version of Problem \ref{GSSM} was
considered in \cite{youssef2015LAA}. If we take $A=0$, then Problem \ref{GSSM} is
reduced to Problem \ref{SSM}.  Hence, the results presented in this paper also
present a solution to Problem \ref{SSM}. We next state the main result of this paper.
For convenience,  throughout this paper, we set
\begin{equation}
\label{gamma}
\Gamma(m,n,k,r):=\frac{m^2}{\big(\sqrt{(k+1)(m-n+r)}-\sqrt{(n-r)(m-k-1)}\big)^2},
\end{equation}
where   $m,n,k,r\in {\mathbb Z}$.
 We have the following result for Problem \ref{GSSM}.

\begin{theorem}
\label{fix-selection} Suppose that $A\in\mathbb{R}^{n\times \ell}$ and
$B\in\mathbb{R}^{n\times m}$ with $\mbox{rank}(A)=r$ and $rank\big([A \  B]\big)=n$.
Then for any fixed $k\in [n-r, m-1]$, there exists a subset
$\mathcal{S}_0\subseteq[m]$ with  cardinality $k$ such that $[A \ B_{\mathcal{S}_0}]$
is full-rank and
\begin{equation}
\label{main-result-bound}
\big\|\big([A \ B_{\mathcal{S}_0}]\big)^{\dagger}\big\|_{\xi}^2 \leq \Gamma(m,n,k,r)\cdot \bigg(1+\frac{\|A^{\dagger}B\|^2_F}{m-n+r}\bigg) \cdot\big\|([A \ \, B])^{\dagger}\big\|_{\xi}^2,
\end{equation}
where  $\xi\in \{2,F\}$.
\end{theorem}

The proof of Theorem \ref{fix-selection} provides a deterministic   algorithm for
computing the subset $\mathcal{S}_0$ in time $O(k(m-\frac{k}{2})n^{\theta}\log
(1/\epsilon))$,  where $\theta>2$ is the exponent of the complexity of matrix
multiplication, which we will introduce it in Section 4.

Taking $A=0$ in Theorem \ref{fix-selection},  we obtain the following corollary.
\begin{corollary}
\label{co:selection} Suppose that $B\in\mathbb{R}^{n\times m}$ with
$rank\big(B\big)=n$. Then for any fixed $k\in [n, m-1]$, there exists a subset
$\mathcal{S}_0\subseteq[m]$ with cardinality $k$ such that
$rank(B_{\mathcal{S}_0})=n$ and for both $\xi=2$ and $F$,
\begin{equation}\label{eq:cobound}
\big\|B_{\mathcal{S}_0}^{\dagger}\big\|_{\xi}^2 \leq \Gamma(m,n,k,0)\cdot\big\|B^{\dagger}\big\|_{\xi}^2.
\end{equation}
\end{corollary}

\subsection{Related work}
\label{related-work} In this subsection, we give a summary of known results on the
subset selection problem and also provide comparisons between our results and those
of previous studies.
\subsubsection{Lower bounds}

A \emph{lower bound} is defined as a non-negative number $\gamma$ such that there
exists a matrix $B\in\mathbb{R}^{n\times m}$ satisfying
$$
\|B^{\dagger}_{\cS}\|_{\xi}^2\geq \gamma \|B^{\dagger}\|^2_{\xi}
$$
for every $\cS$ of cardinality $k\geq n$. Lower bounds for Problem \ref{SSM} were
studied  in \cite{avron2013}. Particularly, for $\xi=2$, Theorem $4.3$ in
\cite{avron2013} showed a lower bound is $\frac{m}{k}-1$. For $\xi=F$, according to
Theorem $4.5$ in \cite{avron2013}, we know  a bound is $\frac{m}{k}-O(1)$ provided
$k=O(n)$. The approximation bound presented in Corollary \ref{co:selection}
asymptotically matches those bounds. Indeed, according to Corollary
\ref{co:selection}, we have
$$
\|B_{\mathcal{S}}^{\dagger}\|_{\xi}^2\leq \Gamma(m,n,k,0) \cdot \|B^{\dagger}\|_{\xi}^2.
$$
If $n/k$ is fixed, then $\Gamma(m,n,k,0)=O(m/k)$, which asymptotically matches  the
lower bounds presented in \cite{avron2013}. Besides, if $k/m$ is fixed and $m$ is
large enough, then $\Gamma(m,n,k,0)\approx m/k$ which is close to the lower bound
$m/k-1$.


\subsubsection{Restricted invertibility principle }
The \emph{restricted invertibility} problem asks whether  one can select a large number of \emph{linearly independent} columns of $B$ and provide an estimation for the norm of the restricted inverse. To be more precise, one wants to find a subset $\cS$, with cardinality $k\leq \mbox{rank}(B)$ being as large as possible, such that $\|B_{\cS}x\|_2\geq c\|x\|_2$ for all $x\in\mathbb{R}^{|\cS|}$ and to estimate the constant $c$.
In \cite{bourgain1987}, Bourgain and Tzafriri introduced the restricted invertibility problem
 and showed its applications in geometry and analysis. Later, their results were improved in  \cite{vershynin2001, youssef2014,youssef2014-2, spielman2012an}.
In \cite{spielman17}, Marcus, Spielman, and Srivastava employed the method of
interlacing families of polynomials  to sharpen this result and presented  a simple
proof to the restricted invertibility principle. One can see \cite{Naor2017} for a
survey of the recent development in restricted invertibility.

Problem \ref{SSM} is different from the restricted invertibility problem. In Problem
\ref{SSM}, we require $|{\mathcal S}|\geq {\rm rank}(B)$, while in  the restricted
invertibility problem, one focuses on the case where $|{\mathcal S}|\leq {\rm
rank}(B)$. We would like to mention that our proof for Theorem \ref{fix-selection} is
inspired by the method developed  by    Marcus, Spielman, and Srivastava
\cite{spielman17} to study the restricted invertibility principle. We will introduce
the main idea of the proof in Section 1.4.

\subsubsection{Approximation bounds for $\xi=F$}
We first focus on known bounds for
$$
\|B_{\mathcal{S}}^{\dagger}\|_{F}^2/\|B^{\dagger}\|_{F}^2.
$$

In \cite{avron2013,hoog2007,hoog2011},  a greedy  algorithm was developed, where one
``bad'' column of $B$ is removed at each step. As shown in
\cite{avron2013,hoog2007,hoog2011},  the greedy algorithm can find a subset $\cS$
with $\abs{\cS}=k\geq n$ such that
\begin{equation}\label{eq:abound}
\|B_{\mathcal{S}}^{\dagger}\|_{F}^2\leq \frac{m-n+1}{k-n+1}\|B^{\dagger}\|_{F}^2
\end{equation}
 in $O(mn^2+mn(m-k))$ time. If $n/k<1$ is
fixed, the  bound $\frac{m-n+1}{k-n+1}$ in (\ref{eq:abound}) is $O(m/k)$, which is as
same as that in (\ref{eq:cobound}).

In \cite{youssef2015LAA}, the Frobenius norm version of Problem \ref{GSSM} was
studied. Let $A$ be a fixed matrix. The author of  \cite{youssef2015LAA} showed that
for any sampling parameter $k\in [ n-\lfloor\|([A \ \,
B])^{\dagger}A\|^2_F\rfloor,m-1]$, one can produce a subset ${\mathcal S}$ satisfying
$\abs{\cS}=k$ in $O((n^3+mn^2)(m-k))$ time while presenting  an upper bound on
$\big\|\big([A \ B_{\mathcal{S}}]\big)^{\dagger}\big\|_{F}^2$.
  Note that Theorem
\ref{fix-selection} requires  the sampling parameter $k\in [n-r,m-1]$, and hence it
is available for a wider range of $k$. Here, we use the fact that $\|([A \ \,
B])^{\dagger}A\|^2_F=\mbox{Tr}\big((AA^T+BB^T)^{-1}AA^T\big)\leq r$, and hence
$n-r\leq n-\lfloor\|([A \ \, B])^{\dagger}A\|^2_F\rfloor$.

\subsubsection{Approximation bounds for $\xi=2$}
For $\xi=2$,   an algorithm  was developed in \cite{avron2013}, which outputs $\cS$
satisfying $\abs{\cS}=k$ and
\begin{equation}\label{eq:2normbound}
\|B_{\mathcal{S}}^{\dagger}\|_{2}^2\leq \left(1+\frac{n(m-k)}{k-n+1}\right)\|B^{\dagger}\|_{2}^2.
\end{equation}
 The algorithm runs in $O(mn^2+mn(m-k))$ time.
If $n/k$ is fixed, the asymptotic  bound in (\ref{eq:2normbound}) is $O(m-k+1)$,
which is larger than that in (\ref{eq:cobound}), i.e.,  $\Gamma(m,n,k,0)$. For
$\xi=2$, to our knowledge, Problem \ref{GSSM} has not been considered in any previous
papers, and Theorem \ref{fix-selection} is the first work on an approximation bound
as well as a deterministic algorithm for Problem \ref{GSSM}.

\subsubsection{Approximation bounds for both $\xi=2$ and $F$}

In \cite{avron2013}, a deterministic algorithm have developed  for both $\xi=2$ and
$F$. The algorithm, which runs in  $O(kmn^2)$ time, outputs a set ${\mathcal S}$
 satisfying $\abs{\mathcal S}=k>n$ and
\begin{equation}\label{eq:F2}
\|B_{\mathcal{S}}^{\dagger}\|_{\xi}^2\leq\left(1+\sqrt{\frac{m}{k}}\right)^2\left(1-\sqrt{\frac{n}{k}}\right)^{-2}\|B^{\dagger}\|_{\xi}^2, \quad \xi\in \{2, F\}.
\end{equation}
Noting that
\begin{eqnarray*}
\Gamma(m,n,k,0)&\leq& \bigg(1+\frac{\sqrt{n(k+1)}}{m}\bigg)^{-1}\frac{m}{(\sqrt{k+1}-\sqrt{n})^2}\\
&<&\frac{m}{(\sqrt{k+1}-\sqrt{n})^2}<\frac{(\sqrt{k}+\sqrt{m})^2}{(\sqrt{k}-\sqrt{n})^2}\\
&=&\bigg(1+\sqrt{\frac{m}{k}}\bigg)^2\bigg(1-\sqrt{\frac{n}{k}}\bigg)^{-2},
\end{eqnarray*}
 we obtain that
$$
\Gamma(m,n,k,0)< \bigg(1+\sqrt{\frac{m}{k}}\bigg)^2\bigg(1-\sqrt{\frac{n}{k}}\bigg)^{-2}.
$$
Hence the bound $\Gamma(m,n,k,0)$, which is presented in Corollary
\ref{co:selection}, is much better than the one in (\ref{eq:F2}). Particularly, when $k$
tends to $n$, the approximation bound in (\ref{eq:F2}) goes to infinity while
$\Gamma(m,n,k,0)$ is still finite. Hence, the bound $\Gamma(m,n,k,0) $ is far better
than the one in (\ref{eq:F2}) when $k$ is close to $n$.

\subsubsection{Algorithms}
Many random algorithms have developed for solving Problem \ref{SSM} (see
\cite{avron2013}). In this paper, we focus on deterministic algorithms. Motivated by
the proof of  Theorem \ref{fix-selection}, we introduce a deterministic algorithm in
Section 4, which outputs a subset ${\mathcal S}$ such that
$$
\big\|\big([A \ B_{\mathcal{S}}]\big)^{\dagger}\big\|_{\xi}^2 \leq \Gamma(m,n,k,r)\cdot\bigg(1+\frac{\|A^{\dagger}B\|^2_F}{m-n+r}\bigg)\cdot\left(1+2k\epsilon\right) \cdot\big\|([A \ \, B])^{\dagger}\big\|_{\xi}^2
$$
for any fixed $\epsilon\in (0,\frac{1}{2k})$.
 As shown in Theorem \ref{th:alg}, the complexity of the algorithm is
 $O(k(m-\frac{k}{2})n^{\theta}\log (1/\epsilon))$
 where $\theta>2$ is the exponent of  matrix multiplication.
We emphasize that our algorithm is faster than all  algorithms mentioned in Section
1.3.3 and Section 1.3.4 when $m$ is large enough because  there exists a factor $m^2$
in the computational cost of all of them, while the time complexity of our algorithm
is linear in $m$.

Note that the time complexity of the algorithm mentioned in Section 1.3.5 is much
better than that of our algorithm. However, as said before, the  approximation bound
obtained by our algorithm is far better than that provided by the algorithms  in
Section 1.3.5. Moreover, our algorithm can solve both Problem \ref{SSM} and Problem
\ref{GSSM}, while all  the other algorithms only work for Problem \ref{SSM}.

\subsection{Our techniques}
Our proof of Theorem \ref{fix-selection} builds on the  method of interlacing
families, which is a powerful technology  developed in \cite{spielman151,spielman152}
(see also \cite{spielman17,spielmaniv})  by Marcus, Spielman, and Srivastava.  Recall
that an interlacing family of polynomials  always contains a polynomial whose $k$-th
largest root is at least the $k$-th largest root of the sum of the polynomials in the
family. This property plays  a key role in our argument.

 Suppose that $Y\in\mathbb{R}^{n\times (m+\ell)}$ whose rows are composed of
  right singular vectors of $[A \ \, B]\in \mathbb{R}^{n\times(m+\ell)}$.
Because  the right singular vectors are orthonormal, we have $YY^T=I$. Our method is
based on the observation that the column space of $[A \ \, B]$ and the column space
of $Y\in\mathbb{R}^{n\times (m+\ell)}$ are identical.
Hence,  we just need to consider the subset selection problem for  the isotropic case
$YY^T=I$ (see Section 3 for detail). We assume that $M$ is a fixed sub-matrix of $Y$
corresponding to a  subset $\cS_M$, i.e., $M=Y_{\cS_M}$. We then show that the
characteristic polynomials of $Y_{\cS}Y_{\cS}^T+MM^T,
\abs{S}=k,\cS\cap\cS_{M}=\emptyset $ form an interlacing family (see Theorem
\ref{interlacing-families}). This implies that there exists  a subset $\cS_0$ such
that the smallest root of the characteristic polynomial of
$Y_{\cS_0}Y_{\cS_0}^T+MM^T$ is at least the smallest root of the \emph{expected
characteristic polynomial}, which is the certain sums of those characteristic
polynomials. Then, we need to present a  lower bound on the smallest root of the
expected characteristic polynomial. We do this by employing the method of the lower
barrier function argument \cite{BSS14,spielman2012an,spielman152,MSS2015FF}. Last but
not the least,  one can use a more generic way provided by the framework of
polynomial convolutions \cite{MSS2015FF} to establish the lower bound here.


\subsection{Organization}
The remainder of the paper is organized as follows. After introducing some
preliminaries in Section 2, we present  the proof of Theorem \ref{fix-selection} in
Section \ref{section-proof}.
 In Section  \ref{section-algorithm}, we finally provide  a deterministic    selection algorithm for computing the subset ${\mathcal S}_0$
  in   Theorem \ref{fix-selection}.

\section{Preliminaries}
\label{section-pre}
\subsection{Notations and Lemmas}
\label{section-pre1}

We use $\partial_{x}$ to denote the operator that performs differentiation with respect to $x$.
We say that a univariate polynomial is \emph{real-rooted} if all of its  coefficients
and roots are real. For a real-rooted polynomial $p$, we let $\lambda_{\min}(p)$ and
$\lambda_{\max}(p)$ denote the smallest and  largest root of $p$, respectively. We
use $\lambda_{k}(p)$ to denote the $k$th largest root of $p$. Let $\cS$ and
$\mathcal{K}$ be two sets; we use $\cS\setminus\mathcal{K}$ to denote the set of
elements in $\cS$ but not in $\mathcal{K}$. We use $\mathbb{E}$  to denote the
expectation of a random variable.

{\bf Singular Value Decomposition.} For a matrix $Q\in\mathbb{R}^{n\times m}$, we
denote the operator  norm and  Frobenius norm of $Q$ by $\|Q\|_2$ and $\|Q\|_F$,
respectively. The (thin) singular value decomposition (SVD) of
$Q\in\mathbb{R}^{n\times m}$ of rank $r=\mbox{rank}(Q)$ is $Q=U\Sigma V^T$, where
$\Sigma={\rm diag}(\sigma_1(Q),\ldots,\sigma_{r}(Q))\in \R^{r\times r}, U\in
\R^{n\times r}, V\in \R^{m\times r}$ such that $U^TU=I, V^TV=I$. For convenience, we
shall repeatedly use the column representation for the matrix $V$, i.e., $Y:=V^T$.
 The $\sigma_1(Q)\geq
\sigma_2(Q)\geq \ldots\geq \sigma_r(Q)$ are known as the singular values of $Q$. The
columns of $U$ and  columns of $V$ are called the left-singular vectors and
right-singular vectors of $Q$, respectively.
   A simple observation is  that $\|Q\|_2=\sigma_{1}(Q)$ and $\|Q\|_F=\sqrt{\sum\limits_{i=1}^r \sigma_i(Q)^2}$.

{\bf Moore-Penrose pseudo-inverse.} Suppose that $Q\in\mathbb{R}^{n\times m}$  and
its thin SVD is $Q=U\Sigma V^T$. We write $Q^{\dagger}=V\Sigma^{-1}U^T\in
\mathbb{R}^{m\times n}$ as the Moore-Penrose pseudo-inverse of $Q$, where
$\Sigma^{-1}$ is the inverse of $\Sigma$. It has the following properties.

\begin{lemma}[\cite{Berstein205}, Fact $6.4.12$]
\label{lem-subset}
Let $P\in \mathbb{R}^{m\times n}$ and $Q\in\mathbb{R}^{n\times \ell}$. If $\rank(P)=\rank(Q)=n$ or $QQ^T=I$, then $(PQ)^{\dagger}=Q^{\dagger}P^{\dagger}$.
\end{lemma}

In general, $(PQ)^{\dagger}\neq Q^{\dagger}P^{\dagger}$ if $Q$ is not full rank. However, if $P$
is a nonsingular square matrix, the following lemma shows that $\|(PQ)^{\dagger}\|_F\leq\|Q^{\dagger}P^{-1}\|_F$. Lemma \ref{lem-moorep} is useful in our argument and we believe that it is of independent interest.
\begin{lemma}
\label{lem-moorep}
Let $P\in \mathbb{R}^{n\times n}$ be an invertible matrix. Then for any $Q\in\mathbb{R}^{n\times \ell}$,  $\|(PQ)^{\dagger}\|_F\leq\|Q^{\dagger}P^{-1}\|_F$.
\end{lemma}
\begin{proof}
Set $J:=PQ$. Then $Q=P^{-1}J$. It suffices to prove
\begin{equation}\label{eq:budeng}
\|J^{\dagger}\|_F^2\leq \|(P^{-1}J)^{\dagger} P^{-1}\|_F^2.
 \end{equation}
 Let $J=U\Sigma V^T$ be the singular value decomposition of $J$,  where $U\in\mathbb{R}^{n\times n}$ and $V\in\mathbb{R}^{\ell\times \ell}$ are two unitary matrices,
 and
$$
\Sigma=\left(
         \begin{array}{cc}
           D_r & 0 \\
           0 & 0\\
         \end{array}
       \right)\in \R^{n\times \ell}
$$
with $D_r=\mbox{diag}\big(\sigma_1(J),\ldots,\sigma_r(J)\big)$ and $r=\rank(J)$.
 Note that
\begin{equation*}\label{eq:deng2}
\begin{array}{ll}
\|(P^{-1}J)^{\dagger} P^{-1}\|_F&=\|(P^{-1}U\Sigma V^T)^{\dagger} P^{-1}\|_F=\|V(P^{-1}U\Sigma)^{\dagger} P^{-1}\|_F
\\
&=\|(P^{-1}U\Sigma)^{\dagger} P^{-1}\|_F=\|(P^{-1}U\Sigma)^{\dagger} P^{-1}U\|_F
\end{array}
\end{equation*}
and  $\|J^{\dagger}\|_F=\|\Sigma^{\dagger}\|_F$. To prove (\ref{eq:budeng}), it is
sufficient to show that
\begin{equation}\label{eq:bd}
\|\Sigma^{\dagger}\|_F\leq \|(P^{-1}U\Sigma)^{\dagger} P^{-1}U\|_F.
\end{equation}
 We use $\ve_j,j=1,\ldots, n,$ to denote a vector in $\R^n$
  whose $j$th entry is $1$ and other entries are $0$. Because  $P^{-1}U$ is invertible,
 the linear systems $\Sigma x=\ve_j$ and $P^{-1}U\Sigma x=P^{-1}U \ve_j$ have the
same solutions. Hence $\|\Sigma^{\dagger}\ve_j\|_2^2=\|(P^{-1}U\Sigma)^{\dagger}
P^{-1}U \ve_j\|_2^2$ for $j=1,\ldots,r$. This implies
$$
\begin{array}{ll}
\|\Sigma^{\dagger}\|_F^2=\sum\limits_{j=1}^r\|\Sigma^{\dagger}\ve_j\|_2^2&=\sum\limits_{j=1}^r\|(P^{-1}U\Sigma)^{\dagger} P^{-1}U \ve_j\|_2^2
\\
&\leq \sum\limits_{j=1}^n\|(P^{-1}U\Sigma)^{\dagger} P^{-1}U \ve_j\|_2^2
\\
&=\|(P^{-1}U\Sigma)^{\dagger} P^{-1}\|_F^2,
\end{array}
$$
and we arrive at  (\ref{eq:bd}) and hence (\ref{eq:budeng}).
\end{proof}

{\bf Jacobi's formula and Jensen's inequality.}
\begin{lemma}[Jacobi's formula]
\label{lemma-jacobi}
Let $P$ and $Q$ be two square matrices. Then,
$$
\partial_x\det[xP+Q]=\det[xP+Q]\cdot \mbox{{\rm Tr}}[P(xP+Q)^{-1}].
$$
\end{lemma}

We next introduce Jensen's inequality.
\begin{lemma}[Jensen's inequality]
\label{lemma-Jensen} Let $f$ be a function from $\mathbb{R}^n$ to
$(-\infty,+\infty]$. Then $f$ is concave if and only if
$$
\mu_1f(x_1)+\cdots+\mu_mf(x_m)\leq f(\mu_1x_1+\cdots+\mu_mx_m)
$$
whenever $\mu_1\geq0,\ldots,\mu_m\geq 0, \mu_1+\cdots+\mu_m=1$.
\end{lemma}

We also need the following lemma.

\begin{lemma}[\cite{Berstein205}, Fact $2.16.3$]
\label{lemma-invertible} If $Q\in {\mathbb R}^{n\times n}$ is an invertible matrix,
then for any vector $u\in {\mathbb R}^n$,
$$
\det[Q+uu^T]=\det[Q](1+u^TQ^{-1}u)=\det[Q](1+\mbox{{\rm Tr}}[Q^{-1}uu^T]).
$$
\end{lemma}


\subsection{Interlacing Families}
\label{section-pre2}
Our proof of Theorem \ref{fix-selection} builds on the method  of interlacing
families which is a powerful  technique developed in \cite{spielman151,spielman152}
by Marcus, Spielman, and Srivastava.

Let $g(x)=\alpha_0 \prod\limits_{i=1}^{n-1}(x-\alpha_i)$ and $f(x)=\beta_0\prod\limits_{i=1}^{n}(x-\beta_i)$ be two real-rooted polynomials. We say $g$ interlaces $f$ if
$$
\beta_1\leq \alpha_1\leq\beta_2\leq\alpha_2\leq \cdots\leq\alpha_{n-1}\leq \beta_n.
$$
We say that polynomials $f_1,\ldots,f_k$ have a \emph{common interlacing} if there is
a polynomial $g$ such that $g$ interlaces $f_i$ for each $i\in \{1,\ldots,k\}$. The
following lemma shows that the common interlacings are  equivalent to the
real-rootedness of convex combinations.
\begin{lemma}[\cite{chudnovsky2007}, Theorem $3.6$]
\label{interlacing-convex} Let $f_1,\ldots,f_m$ be real-rooted (univariate)
polynomials of  the same degree with positive leading coefficients. Then
$f_1,\ldots,f_m$ have a common interlacing if and only if
$\sum\limits_{i=1}^m\mu_if_i$ is real-rooted for all convex combinations $\mu_i\geq
0,\sum\limits_{i=1}^m\mu_i=1$.
\end{lemma}

The following lemma is also useful in our argument.
\begin{lemma}[\cite{spielman17}, Claim $2.9$]
\label{lemma-interlacing} If $Q\in \mathbb{R}^{n \times n}$ is a symmetric matrix and
$u_1,\ldots,u_m$ are vectors in $\mathbb{R}^n$, then the polynomials
$$
f_{j}(x)=\det\big[ x I-Q-u_ju_j^T \big],\quad j=1,\ldots,m
$$
have a common interlacing.
\end{lemma}

Following \cite{spielman17}, we define the notion of an interlacing family of polynomials as follows.
\begin{definition}[\cite{spielman17}, Definition 2.5]
\label{def-inter} An {\em interlacing family} consists of a finite rooted tree
$\mathbb{T}$ and a labeling of the nodes $v\in \mathbb{T}$ by monic real-rooted
polynomials $f_{v}(x)\in\mathbb{R}[x]$, with two properties:
\begin{itemize}
  \item[(a)] Every polynomial $f_{v}(x)$ corresponding to a non-leaf node $v$ is a convex combination of the polynomials corresponding to the children of $v$.
  \item[(b)] For all nodes $v_1,v_2\in \mathbb{T}$ with a common parent, the
      polynomials $f_{v_1}(x),f_{v_2}(x)$ have a common interlacing.
      \footnote{This condition is equivalent to  all convex combinations of all
      the children of a node being real-rooted;
  the equivalence is implied by Helly's theorem and Lemma \ref{interlacing-convex}. }
\end{itemize}
We say that a set of polynomials form an interlacing family if they are the labels of
the leaves of ${\mathbb T}$.
\end{definition}

\begin{figure}[H]
\centering
\includegraphics[width=0.55\textwidth]{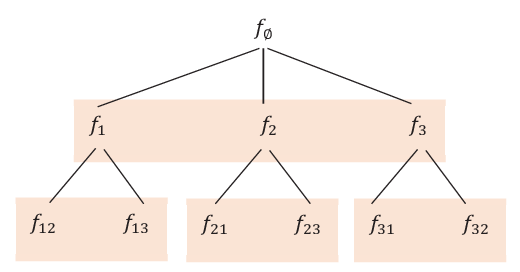}
\caption{A finite rooted tree with $f_{\emptyset}$ as its root. The orange blocks denote subsets of polynomials that have a common interlacing. For every fixed $i$ $(i=\emptyset, 1,2$ or $3)$, each polynomial $f_{i}$ is a convex combination of the polynomials $\{f_{ij}\}_{j\neq i}$. The polynomials $\{f_{ij}\}_{1\leq i\neq j\leq 3}$, which are the labels of the leaves of this tree, form an interlacing family. }
\label{figure:1}
\end{figure}

%

The following lemma, which was proved in \cite[Theorem $2.7$]{spielman17}, shows the
utility of forming an interlacing family.
\begin{lemma}[\cite{spielman17}, Theorem $2.7$]
\label{root-interlacing} Let $f$ be an interlacing family of degree $n$ polynomials
with root labeled by $f_{\emptyset}(x)$ and leaves by $\{f_{v}(x)\}_{v\in
\mathbb{T}}$. Then for all indices $1\leq j\leq n$, there exist leaves $v_1$ and
$v_2$ such that
$$
\lambda_{j}(f_{v_1})\geq \lambda_{j}(f_{\emptyset})\geq \lambda_{j}(f_{v_2}).
$$
\end{lemma}

\subsection{Lower barrier function}
\label{section-pre3} In this section, we introduce the  lower barrier  function from
\cite{BSS14,spielman152}. For a real-rooted polynomial $p(x)$, one can use the
evolution of such a barrier function to track the approximate locations of the roots
of $\partial_x^kp(x)$.

\begin{definition}
For a real-rooted polynomial $p(x)$ with roots $\lambda_1,\ldots,\lambda_n$,   the
lower barrier function of $p(x)$ is defined as
$$
\Phi_{p}(x):=-\frac{p'(x)}{p(x)}=\sum\limits_{i=1}^n\frac{1}{\lambda_i-x}.
$$
\end{definition}

We have the following technical lemma for the lower barrier function, which can be
obtained by  Lemma $5.11$ in \cite{spielman152}.  Here, we include a novel proof for
completeness.

\begin{lemma}
\label{low-descent} Suppose that $p(x)$ is a real-rooted polynomial and $\delta>0$.
Suppose that $b< \lambda_{\min}(p)$ and
$$
\Phi_{p}(b)\leq \frac{1}{\delta}.
$$
Then
\begin{equation}
\label{low-des-r}
\Phi_{\partial_x p}(b+\delta)\leq \Phi_{p}(b).
\end{equation}
\end{lemma}
\begin{proof}
Suppose that the degree of $p$ is $d$ and its roots are
$\lambda_d\leq\cdots\leq\lambda_1$. According to the definition of $\Phi_p$, we have
$$
\frac{1}{\lambda_d-b}<\Phi_{p}(b)\leq \frac{1}{\delta},
$$
 which implies  $b+\delta<\lambda_d\leq \lambda_{\min}(\partial_x p)$.
 Here, we use $\lambda_d\leq \lambda_{\min}(\partial_x p)$, which can be obtained by
 Rolle's theorem.
Next, we  express $\Phi_{\partial_x p}$ in terms of $\Phi_{p}$ and $(\Phi_{p})'$:
$$
\Phi_{\partial_x p}=-\frac{p''}{p'}=-\frac{p''p-(p')^2}{p'p}-\frac{p'}{p}=-\frac{(\Phi_{p})^{'}}{\Phi_{p}}+\Phi_{p}.
$$
wherever all quantities are finite, which happens everywhere except  at the zeros of
$p$ and $p'$. Because $b+\delta$ is strictly below the zeros of both, it follows
that:
$$
\Phi_{\partial_x p}(b+\delta)=\Phi_{p}(b+\delta)-\frac{(\Phi_{p})'(b+\delta)}{\Phi_{p}(b+\delta)}.
$$
Therefore, \eqref{low-des-r} is equivalent to
$$
\Phi_{p}(b+\delta)-\Phi_{p}(b)\leq\frac{(\Phi_{p})'(b+\delta)}{\Phi_{p}(b+\delta)},
$$
i.e.,
$$
\Phi_{p}(b+\delta)\big( \Phi_{p}(b+\delta)-\Phi_{p}(b)\big)\leq(\Phi_{p})'(b+\delta).
$$
By expanding $\Phi_{p}$ and $(\Phi_{p})'$ in terms of the zeros of $p$,  we can see
that  \eqref{low-des-r} is equivalent to
\begin{equation}\label{eq:9t}
\bigg(\sum\limits_{i}\frac{1}{\lambda_i-b-\delta}\bigg)\bigg(\sum\limits_{i}\frac{1}{\lambda_i-b-\delta}-\sum\limits_{i}\frac{1}{\lambda_i-b}\bigg)\leq
\sum\limits_{i}\frac{1}{(\lambda_i-b-\delta)^2}.
\end{equation}
Noting that
$\frac{1}{\lambda_i-b-\delta}=\frac{1}{\lambda_i-b}+\frac{\delta}{(\lambda_i-b)(\lambda_i-b-\delta)}$,
we obtain that
$$
\begin{array}{ll}
\bigg(\sum\limits_{i}\frac{1}{\lambda_i-b-\delta}\bigg)\bigg(\sum\limits_{i}\frac{1}{\lambda_i-b-\delta}&-\sum\limits_{i}\frac{1}{\lambda_i-b}\bigg)
=\bigg(\sum\limits_{i}\frac{1}{\lambda_i-b}+\sum\limits_{i}\frac{\delta}{(\lambda_i-b)(\lambda_i-b-\delta)}\bigg)\bigg(\sum\limits_{i}\frac{\delta}{(\lambda_i-b)(\lambda_i-b-\delta)}\bigg)
\\
&=\bigg(\sum\limits_{i}\frac{\delta}{(\lambda_i-b)(\lambda_i-b-\delta)}\bigg)^2+\bigg(\sum\limits_{i}\frac{\delta}{\lambda_i-b}\bigg)\bigg(\sum\limits_{i}\frac{1}{(\lambda_i-b)(\lambda_i-b-\delta)}\bigg)
\\
&\leq\bigg(\sum\limits_{i}\frac{\delta}{\lambda_i-b}\bigg)\bigg(\sum\limits_{i}\frac{\delta}{(\lambda_i-b)(\lambda_i-b-\delta)^2}\bigg)+\sum\limits_{i}\frac{1}{(\lambda_i-b)(\lambda_i-b-\delta)}
\\
&\leq \sum\limits_{i}\frac{\delta}{(\lambda_i-b)(\lambda_i-b-\delta)^2}+\sum\limits_{i}\frac{1}{(\lambda_i-b)(\lambda_i-b-\delta)}
\\
&=\sum\limits_{i}\frac{1}{(\lambda_i-b-\delta)^2}-\sum\limits_{i}\frac{1}{(\lambda_i-b)(\lambda_i-b-\delta)}+\sum\limits_{i}\frac{1}{(\lambda_i-b)(\lambda_i-b-\delta)}
\\
&=\sum\limits_{i}\frac{1}{(\lambda_i-b-\delta)^2},
\end{array}
$$
which implies (\ref{eq:9t}) and hence (\ref{low-des-r}). Here, the first and  second
inequalities follow from $\Phi_{p}(b)\leq \frac{1}{\delta}$, i.e.,
$\sum\limits_{i}\frac{\delta}{\lambda_i-b}\leq 1$ and the Cauchy-Schwarz inequality,
respectively.
\end{proof}


\section{Proof of Theorem \ref{fix-selection}}
\label{section-proof}

In this section, we present the proof of  Theorem \ref{fix-selection}.  Our proof
provides a deterministic greedy algorithm that will be presented in Section
\ref{section-algorithm}. To state our proof, we introduce the following result but
postpone   its proof till the end of this section.
\begin{theorem}
\label{fix-theorem-r1} Let $Y\in {\mathbb R}^{n\times(m+\ell)}$ that satisfies
$YY^T=I$. Assume that $\cS_{M}\subseteq[m+\ell]$ with $|\cS_{M}|=\ell$. Let
$M:=Y_{\cS_M}\in\mathbb{R}^{n\times \ell}$ be a sub-matrix of $Y$ whose columns are
indexed by $\cS_{M}$. Set $r:=rank(M)$. Then for any fixed $k\in [n-r, m-1]$ there
exists a subset $\cS_0\subseteq[m+\ell]\setminus \cS_{M}$ with size $k$ such that
$$
\sigma_{\min}\big(\big[M \ Y_{\cS_0}\big]\big)^2\geq \Gamma^{-1}(m,n,k,r)\cdot\frac{m-n+r}{m-n+\|M^{\dagger}\|_F^2},
$$
where $\Gamma(m,n,k,r)$ is defined in \eqref{gamma}.
\end{theorem}

Using  Theorem \ref{fix-theorem-r1}, we next present the proof of Theorem
\ref{fix-selection}.
\begin{proof}[Proof of Theorem \ref{fix-selection}]
Let $[A \ \, B]=U\Sigma Y$ be the SVD of $[A \ \ B]$.
 Suppose that $\cS_A\subseteq [m+\ell]$ and $\cS_B\subseteq [m+\ell]$ are two indexed sets such that
 $A=U\Sigma Y_{\cS_A}$ and $B=U\Sigma Y_{\cS_B}$.

 Recall that $\mbox{rank}\big([A \ B]\big)=n$, which implies that $YY^T=I_n$.  Applying Theorem \ref{fix-theorem-r1}
  with $M:=Y_{\cS_A}$, we obtain that there exists a subset $\mathcal{S}_0\subseteq [m+\ell]\setminus\cS_A$ with size $k\in [n-r, m-1]$ such that
\begin{equation}
\label{main-proof-2}
\sigma_{\min}\big(\big[Y_{\cS_A} \ \, Y_{\cS_0}\big]\big)^2\geq \Gamma^{-1}(m,n,k,r)\cdot\frac{m-n+r}{m-n+\|(Y_{\cS_A})^{\dagger}\|^2_F}.
\end{equation}
Considering the left side of \eqref{main-result-bound}, we have
\begin{equation}
\label{main-proof-1}
\big\|\big(\big[A \ \, B_{\mathcal{S}_0}\big]\big)^{\dagger}\big\|_{\xi}^2=\big\|\big( [U\Sigma Y_{\cS_A} \ \, U\Sigma Y_{\cS_0}]\big)^{\dagger}\big\|_{\xi}^2=\big\|\big(U\Sigma ([Y_{\cS_A} \  \, Y_{\cS_0}])\big)^{\dagger}\big\|_{\xi}^2.
\end{equation}
From \eqref{main-proof-2}, we know  that the matrix $[Y_{\cS_A} \ \, Y_{\cS_0}\big]$
has full row rank. Since $U\Sigma$ also has full column rank, by Lemma
\ref{lem-subset} we know that $\mbox{rank}([A \ \, B_{\mathcal{S}_0}])=n$ and
\begin{equation}
\label{main-proof-3}
\begin{array}{ll}
\big\|\big(U\Sigma ([Y_{\cS_A} \ \, Y_{\cS_0}])\big)^{\dagger}\big\|_{\xi}^2&=\big\|([Y_{\cS_A} \ \, Y_{\cS_0}])^{\dagger}\Sigma^{-1}U^T\big\|_{\xi}^2
\\[2mm]
&\overset{(a)}{\leq} \big\|([Y_{\cS_A} \ \, Y_{\cS_0}])^{\dagger} \big\|_{2}^2\cdot \big\|([A \ \, B])^{\dagger}\big\|_{\xi}^2
\\[2mm]
&\overset{(b)}{\leq}\Gamma(m,n,k,r)\cdot \frac{m-n+\|(Y_{\cS_A})^{\dagger}\|^2_F}{m-n+r}\cdot \big\|([A \ \, B])^{\dagger}\big\|_{\xi}^2,
\end{array}
\end{equation}
where $(a)$ follows from the standard properties of matrix norms and using the
definition of the pseudoinverse of $[A \ \, B]$ and $\Sigma$, and $(b)$ follows from
\eqref{main-proof-2}. To complete the proof, we still need to present an upper bound
on $\|(Y_{\cS_A})^{\dagger}\|^2_F$. Note that
\begin{equation}
\label{proof-main-111}
\begin{array}{ll}
\|(Y_{\cS_A})^{\dagger}\|^2_F\overset{(a)}{=}\|(\Sigma^{-1}U^TA)^{\dagger} \|^2_F & \overset{(b)}{\leq}\|A^{\dagger}U\Sigma \|^2_F
\\
&\overset{(c)}{=}\|A^{\dagger}U\Sigma YY^T \|^2_F
\\
&\overset{(d)}{\leq}\|A^{\dagger}[A \ \, B] \|^2_F
\\
&\overset{(e)}{=}r+\|A^{\dagger}B\|^2_F,
\end{array}
\end{equation}
where $(a)$ follows from $A=U\Sigma Y_{\cS_A}$, $(b)$ follows from Lemma
\ref{lem-moorep}, $(c)$ follows from $YY^T=I$, $(d)$ follows from $U\Sigma Y=[A \ \,
B]$, the standard properties of matrix norms, and $\|Y^T\|_2\leq 1$, and $(e)$
follows from $\mbox{rank}(A)=r$. Thus, combining \eqref{main-proof-1},
\eqref{main-proof-3}, and \eqref{proof-main-111},  we arrive at
\eqref{main-result-bound}.
\end{proof}

The remainder  of this section aims to prove  Theorem \ref{fix-theorem-r1}. The proof
consists of two  parts. We first prove that the characteristic polynomials of the
matrices that arise in Theorem \ref{fix-theorem-r1} form an interlacing family.
Secondly, we use the barrier function argument to establish a lower bound on the
smallest zero of the expected characteristic polynomial.

\subsection{An interlacing family for subset selection}

Let the columns of $Y$ be the vectors $y_1,\ldots,y_{m+\ell}\in \mathbb{R}^n$,  and
let $M=Y_{\cS_M}$ be a given matrix with $\mbox{rank}(M)=r$. Since $YY^T=I$, we
obtain that $\sum\limits_{i=1}^{m+\ell}y_iy_i^T=I$. Denote the nonzero singular
values of $M$ as $\sigma_1(M),\ldots,\sigma_r(M)$. For each
$\cS\subseteq[m+\ell]\setminus\cS_{M}$, set
$$
 p^M_{\cS}(x):=\det[xI-Y_{\cS}Y^T_{\cS}-MM^T].
$$
For a fixed set $\mathcal{T}$ with size less than $k$, we define the polynomial
\begin{equation}\label{eq:deff}
f^M_{\mathcal{T}}(x):=\mathop{\mathbb{E}}\limits_{\cS\supseteq\mathcal{T},|\cS|=k \atop \cS\subseteq[m+\ell]\setminus\cS_{M}}p^M_{\cS}(x),
\end{equation}
where the expectation is taken uniformly over  sets
$\cS\subseteq[m+\ell]\setminus\cS_{M}$ with size $k$ containing $\mathcal{T}$.
Building on the ideas of Marcus-Spielman-Srivastava  \cite{spielman17},  we can
derive expressions for the polynomials $f^M_{\mathcal{T}}(x)$.

 We begin with the following result.

\begin{lemma}
\label{lemma-5}
Suppose that $p^M_{\cS}(x)=\det[xI-Y_{\cS}Y^T_{\cS}-MM^T]$.
Then
$$
\sum\limits_{i\notin \mathcal{S}\cup \cS_{M}}p^M_{\mathcal{S}\cup\{i\}}(x)=(x-1)^{-(m-n-t-1)}\partial_x(x-1)^{m-n-t}p^M_{\mathcal{S}}(x)
$$
holds for every subset $\mathcal{S}\subseteq[m+\ell]\setminus \cS_{M}$ with size $t$.
\end{lemma}
\begin{proof}
According to Lemma \ref{lemma-invertible}, we have
$$
\begin{array}{ll}
\sum\limits_{i\notin \mathcal{S}\cup \cS_{M}}p^M_{\mathcal{S}\cup\{i\}}(x)&=\sum\limits_{i\notin \mathcal{S}\cup \cS_{M}}\det\big[xI-Y_{\cS}Y^T_{\cS}-y_iy_i^T-MM^T\big]
\\[2mm]
&=\sum\limits_{i\notin \mathcal{S}\cup \cS_{M}}\det\big[xI-Y_{\cS}Y^T_{\cS}-MM^T\big]\big(1-\mbox{Tr}\big[(xI-Y_{\cS}Y^T_{\cS}-MM^T)^{-1}y_iy_i^T\big]\big).
\end{array}
$$
Because $p^M_{\mathcal{S}}(x)=\det[xI-Y_{\cS}Y^T_{\cS}-MM^T]$ and
$\sum\limits_{i\notin \mathcal{S}\cup \cS_{M}}y_iy_i^T=I-Y_{\cS}Y^T_{\cS}-MM^T$, we
obtain that
$$
\begin{array}{ll}
\sum\limits_{i\notin \mathcal{S}\cup \cS_{M}}p^M_{\mathcal{S}\cup\{i\}}(x)&=p^M_{\mathcal{S}}(x)\big(m-t-\mbox{Tr}\big[(xI-Y_{\cS}Y^T_{\cS}-MM^T)^{-1}(I-Y_{\cS}Y^T_{\cS}-MM^T)\big]\big)
\\[2mm]
&=p^M_{\mathcal{S}}(x)(m-t)-p^M_{\mathcal{S}}(x)\mbox{Tr}\big[\big(xI-Y_{\cS}Y^T_{\cS}-MM^T\big)^{-1}
\\[1.0mm]
& \ \ \ \times \big((I-xI)+(xI-Y_{\cS}Y^T_{\cS}-MM^T)\big)\big]
\\[2mm]
&=p^M_{\mathcal{S}}(x)(m-t)-np^M_{\mathcal{S}}(x)-p^M_{\mathcal{S}}(x)\mbox{Tr}\big[(xI-Y_{\cS}Y^T_{\cS}-MM^T)^{-1}(I-xI)\big]
\\[2mm]
&=p^M_{\mathcal{S}}(x)(m-t-n)-p^M_{\mathcal{S}}(x)\mbox{Tr}\big[(xI-Y_{\cS}Y^T_{\cS}-MM^T)^{-1}\big](1-x)
\\[2mm]
&\overset{(a)}{=}(m-n-t)p^M_{\mathcal{S}}(x)+(x-1)\partial_xp^M_{\mathcal{S}}(x)
\\[2mm]
&=(x-1)^{-(m-n-t-1)}\partial_x(x-1)^{m-n-t}p^M_{\mathcal{S}}(x),
\end{array}
$$
where $(a)$ follows from Lemma \ref{lemma-jacobi}.
\end{proof}

Motivated by Lemma $5.3$ in \cite{spielman17}, we     give expressions for
$f^M_{\mathcal{T}}(x)$ in the following lemma:
\begin{lemma}
\label{f-expression}
Suppose that  $\mathcal{T}\subseteq[m+\ell]\setminus \cS_{M}$ with size $t\leq k$.
Then
\begin{displaymath}
f_{\mathcal{T}}^M(x)=\frac{1}{{m-t \choose k-t}
}\sum\limits_{\cS\supseteq\mathcal{T},|\cS|=k \atop
\cS\subseteq[m+\ell]\setminus\cS_{M}}p_{\cS}^M(x)=\frac{(m-k)!}{(m-t)!}(x-1)^{-(m-n-k)}\partial_x^{k-t}(x-1)^{m-n-t}p^M_{\mathcal{T}}(x),
\end{displaymath}
where $ p^M_{\mathcal T}(x)=\det[xI-Y_{\mathcal T}Y^T_{\mathcal T}-MM^T]$. In particular,
$$
f^M_{\emptyset}(x)=\frac{(m-k)!}{m!}(x-1)^{-(m-n-k)}\partial_x^{k}(x-1)^{m-n}x^{n-r}\prod\limits_{i=1}^r\big(x-\sigma_i(M)^2\big).
$$
\end{lemma}
\begin{proof}
To prove this lemma, it is sufficient to show that
\begin{equation}\label{eq:induck}
\sum\limits_{\cS\supseteq\mathcal{T},|\cS|=k \atop \cS\subseteq[m+\ell]\setminus\cS_{M}}p^M_{\cS}(x)=\frac{1}{(k-t)!}(x-1)^{-(m-n-k)}\partial_x^{k-t}(x-1)^{m-n-t}p^M_{\mathcal{T}}(x).
\end{equation}
 We prove (\ref{eq:induck}) by induction on $k$.
For $k=t$, a simple calculation shows that (\ref{eq:induck}) holds. For $k>t$, we
have
$$
\sum\limits_{\cS\supseteq\mathcal{T},|\cS|=k \atop \cS\subseteq[m+\ell]\setminus\cS_{M}}p^M_{\cS}(x)=\frac{1}{k-t}\sum\limits_{\cS\supseteq\mathcal{T},|\cS|=k-1 \atop \cS\subseteq[m+\ell]\setminus\cS_{M}}\sum\limits_{i\notin \cS\cup \cS_{M}} p^M_{\cS\cup\{i\}}(x).
$$
Hence by induction and  Lemma \ref{lemma-5}, we have
$$
\begin{array}{ll}
\sum\limits_{\cS\supseteq\mathcal{T},|\cS|=k \atop \cS\subseteq[m+\ell]\setminus\cS_{M}}p^M_{\cS}(x)
&=\frac{1}{k-t} \sum\limits_{\cS\supseteq\mathcal{T},|\cS|=k-1 \atop \cS\subseteq[m+\ell]\setminus\cS_{M}} (x-1)^{-(m-n-(k-1)-1)}\partial_x(x-1)^{m-n-(k-1)}p^M_{\mathcal{S}}(x)
\\[8mm]
&=\frac{1}{k-t}(x-1)^{-(m-n-(k-1)-1)}\partial_x \bigg((x-1)^{m-n-(k-1)} \sum\limits_{\cS\supseteq\mathcal{T},|\cS|=k-1 \atop \cS\subseteq[m+\ell]\setminus\cS_{M}}p^M_{\mathcal{S}}(x)\bigg)
\\[8mm]
&=\frac{1}{k-t}(x-1)^{-(m-n-(k-1)-1)}\partial_x\big((x-1)^{m-n-(k-1)}
\\
&\ \ \ \times \frac{1}{(k-1-t)!}(x-1)^{-(m-n-(k-1))}\partial_x^{k-1-t}(x-1)^{m-n-t}p^M_{\mathcal{T}}(x)\big)
\\[4.5mm]
&=\frac{1}{(k-t)!}(x-1)^{-(m-n-k)}\partial_x^{k-t}(x-1)^{m-n-t}p^M_{\mathcal{T}}(x).
\end{array}
$$
This completes the proof of this lemma.
\end{proof}

When $M=0$, Marcus, Spielman£¬ and Srivastava \cite[Theorem $5.4$]{spielman17}
proved that the polynomials $p_{\cS}(x):=p^{0}_{\cS}(x)$ for $|\cS|=k$ form an
interlacing family. Inspired by the arguments of Marcus-Spielman-Srivastava in
\cite{spielman17}, we prove that the polynomials $p^M_{\mathcal{\cS}}(x)$ for
$|\cS|=k$ still satisfy the requirements of interlacing families.

\begin{theorem}
\label{interlacing-families} The polynomials
$p^M_{\cS}(x)=\det\big[xI-Y_{\cS}Y^T_{\cS}-MM^T\big]$  for $|\cS|=k$ are an
interlacing family.
\end{theorem}
\begin{proof}

We construct a tree  $\mathbb{T}$ whose  nodes  are
$\mathcal{T}\subseteq[m+\ell]\setminus \cS_{M}$ (possibly empty) with size less than
or equal to  $k$. The node $\mathcal{T}_1$ is a child  of $\mathcal{T}_2$ if and only
if $\mathcal{T}_2\subset \mathcal{T}_1$ and $\abs{\mathcal{T}_1\setminus
\mathcal{T}_2}=1$. For an internal node $\mathcal{T}\subseteq[m+\ell]\setminus
\cS_{M}$ (possibly empty) with size less than $k$, we label  $\mathcal{T}$ by the
polynomial $f^M_{\mathcal{T}}(x)$, which is defined in (\ref{eq:deff}). Similarly, we
label  the leaves of $\mathbb{T}$ as $p^M_{\cS}(x)$ with $\abs{\cS}=k$. Note that the
polynomials $p^M_{\cS}(x)$ are real-rooted and monic, which implies that the
polynomials $f^M_{\mathcal{T}}(x)$ are also monic as they are the averages of
$p^{M}_{\mathcal{S}}(x)$ with $\cS\subseteq[m+\ell]\setminus\cS_{M}$ with size $k$
containing $\mathcal{T}$. We will show that the polynomials $f^M_{\mathcal{T}}(x)$
are  real-rooted later. Now we
 already constructed the finite rooted tree $\mathbb{T}$, where $f^M_{\emptyset}(x)$
is the root of the  tree.

We already showed that the tree $\mathbb{T}$ satisfies  $(a)$ in  Definition
\ref{def-inter}. We next show that $\mathbb{T}$ also satisfies condition $(b)$.

Suppose that $\mathcal{T}\subseteq[m+\ell]\setminus \cS_{M}$ with size less than $k$. 
To complete the proof, by Lemma \ref{lemma-interlacing},  we need to prove that  all
convex combinations of $ f^M_{\mathcal{T}\cup\{i\}} $ and $
f^M_{\mathcal{T}\cup\{j\}} $ are real-rooted for every $i,j\notin
\mathcal{T}\cup\cS_M$.
 That is, we must prove that the polynomial
$$
q_{\mu}(x):=\mu f^M_{\mathcal{T}\cup\{i\}}(x)+(1-\mu) f^M_{\mathcal{T}\cup\{j\}}(x)
$$
is real-rooted  for each $0\leq \mu\leq 1$. Let
$$
h_{\mu}(x):=\mu p^M_{\mathcal{T}\cup\{i\}}(x)+(1-\mu)p^M_{\mathcal{T}\cup\{j\}}(x).
$$
It follows from Lemma \ref{lemma-interlacing} that  the polynomials
$p^M_{\mathcal{T}\cup\{i\}}$ and $p^M_{\mathcal{T}\cup\{j\}}$ have a common
interlacing. Hence, by Lemma \ref{interlacing-convex}, we have that $h_{\mu}(x)$ is
real-rooted. According to  Lemma \ref{f-expression}, we obtain  that
$$
q_{\mu}(x)=\frac{(m-k)!}{(m-t)!}(x-1)^{-(m-n-k)}\partial_x^{k-t}(x-1)^{m-n-t}h_{\mu}(x).
$$
Noting the real rootedness can be  preserved  by multiplication by $(x-1)$, taking
derivatives, and dividing by $(x-1)$ when $1$ is a root,  we can obtain that $
q_{\mu}(x) $ is real-rooted.
\end{proof}

\subsection{Proof of Theorem \ref{fix-theorem-r1}}

The aim of  this subsection is to prove Theorem \ref{fix-theorem-r1}. We first
establish a lower bound on the smallest zero of $f^M_{\emptyset}(x)$  using the lower
barrier function.

\begin{lemma}
\label{lemma-fix-r1}
Suppose that $M\in {\mathbb R}^{n\times \ell}$ with nonzero singular values $\sigma_1,\ldots,\sigma_r\in(0,1]$.
For $n-r\leq k\leq m-1$, let
$$
f^M_{\emptyset}(x)=\frac{(m-k)!}{m!}(x-1)^{-(m-n-k)}\partial_x^{k}(x-1)^{m-n}x^{n-r}\prod
\limits_{i=1}^r\big(x-\sigma_i^2\big).
$$
Then
$$
\lambda_{\min}(f^M_{\emptyset})\geq \Gamma^{-1}(m,n,k,r)\cdot \frac{m-n+r}{m-n+\|M^{\dagger}\|^2_F},
$$
where $\lambda_{\min}(f^M_{\emptyset})$ denotes the smallest zero of
$f^M_{\emptyset}$  and $\Gamma(m,n,k,r)$ is defined in \eqref{gamma}.
\end{lemma}
\begin{proof}
Let
$$
g(x):=\partial_x^{k}(x-1)^{m-n}x^{n-r}\prod\limits_{i=1}^r\big(x-\sigma_i^2\big)
$$
and
$$
p(x):=(x-1)^{m-n}x^{n-r}\prod\limits_{i=1}^r\big(x-\sigma_i^2\big).
$$
By Rolle's theorem, we know that $\partial_xp$ interlaces $p$.  Thus, applying this
fact $k$ times and noting that all the zeros of $p(x)$  belong to $[0,1]$, we
conclude that all the zeros of $g(x)$ are between $0$ and $1$, which implies that
$\lambda_{\min}(f^M_{\emptyset})=\lambda_{\min}(g)$. Thus, it is sufficient to prove
that
$$
\lambda_{\min}( g)\geq \Gamma^{-1}(m,n,k,r)\cdot \frac{m-n+r}{m-n+\sum\limits_{i=1}^r\frac{1}{\sigma_i^2}}.
$$
For convenience, we set
$$
\Lambda:=\Lambda(m,n,r,\sigma):=\frac{m-n+r}{m-n+\sum\limits_{i=1}^r\frac{1}{\sigma_i^2}}.
$$
 For any $\delta>0$, let
$$
b:=b(\delta):=\frac{(\Lambda-m\delta)-\sqrt{(\Lambda-m\delta)^2+4\delta \Lambda(n-r)}}{2}.
$$
Note that $b<\lambda_{\rm min}(p)$.
 We claim that
\begin{equation}
\label{barrier-proof-1}
\Phi_{g}(b+k\delta)\leq \frac{1}{\delta}.
\end{equation}
Indeed, according to the definition of   the lower barrier function of $p$, we obtain
that
\begin{equation}
\label{inequation-p}
\begin{array}{ll}
\Phi_{p}(b)=-\frac{p'}{p}&=\frac{m-n}{-b+1}+\sum\limits_{i=1}^r\frac{1}{-b+\sigma_i^2}+\frac{n-r}{-b}
\\[6mm]
&=(m-n+r)\bigg(\frac{1}{m-n+r}\cdot\frac{m-n}{-b+1}+\frac{1}{m-n+r}\sum\limits_{i=1}^r\frac{1}{-b+\sigma_i^2}\bigg)+\frac{n-r}{-b}
\\[6mm]
&\leq \frac{m-n+r}{-b+\frac{1}{\frac{1}{m-n+r}\big(m-n+\sum\limits_{i=1}^r\frac{1}{\sigma_i^2}\big)}}+\frac{n-r}{-b}=\frac{m-n+r}{-b+\Lambda}+\frac{n-r}{-b}
\\[8mm]
&=\frac{1}{\delta},
\end{array}
\end{equation}
where the inequality follows from Lemma \ref{lemma-Jensen} with   the fact  the
function $x\to \frac{1}{-b+\frac{1}{x}}$  is concave on $(0,+\infty)$. Applying Lemma
\ref{low-descent} $k$ times, we obtain
$$
\Phi_{g}(b+k\delta)\leq \Phi_{p}(b)\leq\frac{1}{\delta},
$$
which implies \eqref{barrier-proof-1}. Noting that
$$
\frac{1}{\lambda_{\min}(g)-(b+k\delta)}\leq \Phi_{g}(b+k\delta)\leq \frac{1}{\delta},
$$
we obtain that
\begin{equation}
\label{equation-min}
\lambda_{\min}(g)\geq b+(k+1)\delta,
\end{equation}
i.e.,
\begin{equation}\label{eq:xiao}
\lambda_{\min}(g)\geq\mu(\delta):=\frac{(\Lambda-m\delta)-\sqrt{(\Lambda-m\delta)^2+4\delta \Lambda(n-r)}}{2}+(k+1)\delta.
\end{equation}
Now we derive the value of $\delta$ at which $\mu(\delta)$ is maximized. Taking derivatives in $\delta$, we obtain that
$$
\mu'(\delta)=k+1-\frac{m}{2}-\frac{m^2\delta-m\Lambda+2\Lambda(n-r)}{2 \sqrt{(\Lambda-m\delta)^2+4\delta \Lambda(n-r)} }.
$$
As $n-r\leq k\leq m-1$, we know that
$$
\mu'(0)=k+1-(n-r)\geq0 \ \  \mbox{and }  \  \  \lim\limits_{\delta\to \infty}\mu'(\delta)=k+1-m\leq 0.
$$
By continuity, a maximum will occur at a point $\delta^*\geq 0$ at which $\mu'(\delta^*)=0$. The solution is given by
$$
\delta^*=\frac{m\Lambda-2\Lambda(n-r)}{m^2}-\frac{m-2(k+1)}{m^2}\sqrt{\frac{\Lambda^2(n-r)(m-n+r)}{(k+1)(m-k-1)}},
$$
which is positive for $m>k\geq n-r$.
Observing that
$$
(\Lambda-m\delta^*)^2+4\delta^*\Lambda(n-r)=\frac{\Lambda^2(n-r)(m-n+r)}{(k+1)(m-k-1)}
$$
and by calculation,
we can obtain that
\begin{equation}\label{eq:xiao1}
\mu(\delta^*)=\Gamma^{-1}(m,n,k,r) \cdot \Lambda.
\end{equation}
Combining (\ref{eq:xiao}) and (\ref{eq:xiao1}), we obtain the lemma.
\end{proof}

Now we are ready to prove Theorem \ref{fix-theorem-r1}.
\begin{proof}[Proof of Theorem \ref{fix-theorem-r1}]
According to Theorem \ref{interlacing-families}, the polynomials $p^M_{\cS}(x)$ with
$|\cS|=k$ form an interlacing family.  Lemma \ref{root-interlacing} implies that
there exists a subset $\cS_0\subseteq[m+\ell]\setminus \cS_M$ such that
$\abs{\cS_0}=k$ and
$$
\lambda_{\min}\big(f^M_{\cS_0}\big)\geq\lambda_{\min}\big(f^M_{\emptyset}\big)\geq \Gamma^{-1}(m,n,k)\cdot\frac{m-n+r}{m-n+\|M^{\dagger}\|^2_F}.
$$
Here, the second inequality follows from Lemma \ref{lemma-fix-r1}.
 As $f^M_{\cS_0}$ is the characteristic polynomial of
$MM^T+Y_{\cS_0}Y_{\cS_0}^T$, we conclude that there exists a
subset $\cS_0\subseteq [m+\ell]\setminus\cS_M$ with size $k$ for which
$$
\sigma_{\min}\big(\big[M \ Y_{\cS_0}\big]\big)^2=\lambda_{\min}\bigg(MM^T+Y_{\cS_0}Y_{\cS_0}^T\bigg)\geq \Gamma^{-1}(m,n,k)\cdot\frac{m-n+r}{m-n+\|M^{\dagger}\|^2_F}.
$$
\end{proof}


\section{A deterministic greedy selection algorithm}
\label{section-algorithm}

The aim of this section is to present a deterministic greedy selection  algorithm for
Problem \ref{GSSM}. The proposed algorithm is based on the proof of the main result.
Suppose that  $A\in\mathbb{R}^{n\times \ell}$ and $B\in\mathbb{R}^{n\times m}$ with
$\mbox{rank}([A \ \ B])=n$. Let the  SVD of $[A \ \, B]$ be $U\Sigma Y$,
 where $Y=[Y_{\cS_A} \ \, Y_{\cS_B}]$ and $\cS_A$ and $\cS_B$ denote the column set of matrices
$A$ and $B$, respectively. Assume that $M:=Y_{\cS_A}$ and $Y=[y_1,\ldots,y_{m+\ell}]$. Given a partial
assignment $\{s_1,\ldots,s_j\}\subseteq\cS_B$, from Lemma \ref{f-expression}, we set
the polynomial corresponding to $\{s_1,\ldots,s_j\}$ as
\begin{equation}
\label{alg-fm}
f^M_{s_1,\ldots,s_j}(x):=\frac{(m-k)!}{(m-j)!}(x-1)^{-(m-n-k)}\partial_x^{k-j}(x-1)^{m-n-j}p^M_{s_1,\ldots,s_j}(x),
\end{equation}
where
$p^M_{s_1,\ldots,s_j}(x)=\det\bigg[xI-\sum\limits_{i=1}^j{y_{s_i}y^T_{s_i}}-MM^T\bigg]$.
The algorithm produces the subset $\mathcal{S}_0$  in polynomial time by iteratively
adding columns to it. Namely, suppose that at the $(i-1)$-th ($1\leq i\leq k$)
iteration, we  have already found a partial assignment $s_1,\ldots,s_{i-1}$ (which is
empty when $i=1$). Now, at the $i$-th iteration, the algorithm finds an index
$s_{i}\in\cS_B\setminus \{s_1,\ldots,s_{i-1}\}$ such that
$\lambda_{\min}(f^M_{s_1,\ldots,s_{i}})\geq
\lambda_{\min}(f^M_{s_1,\ldots,s_{i-1}})$.

Let $p(x)$ be a given real-rooted polynomial; we use
$\lambda^\epsilon_{\min}\big(p(x)\big)$ to denote an $\epsilon$-approximation to the
smallest root of $p(x)$, i.e.,
$$
\big|\lambda^\epsilon_{\min}\big(p(x)\big)-\lambda_{\min}\big(p(x)\big)\big|\leq \epsilon.
$$
The deterministic greedy selection algorithm can be stated as follows:

\begin{algorithm}[H]
\caption{Deterministic greedy selection algorithm}
\begin{algorithmic}[H]
\Require
$B\in\mathbb{R}^{n\times m}$ of rank $n$; $A\in\mathbb{R}^{n\times \ell}$;  sampling parameter $k\in \{n-\mbox{rank}(A),\ldots, m-1\}$.
\begin{enumerate}
\item[1:] Set $s_0=\emptyset$ and $i:=1$.
\item[2:] Compute the thin SVD of $[A \ \, B]=U\Sigma Y$ with $Y=[Y_{\cS_A} \ \, Y_{\cS_B}]$.
\item[3:] Let  $M:=Y_{\cS_A} $ and $Y=[y_1,\ldots,y_{m+\ell}]\in\mathbb{R}^{n\times (m+\ell)}$.
\item[4:] Using the standard technique of binary search with a Sturm sequence, for each
$s \in\cS_B\setminus \{s_1,\ldots,s_{i-1}\}$,
compute  an $\epsilon$-approximation to the smallest root of $f^M_{s_1,\ldots,s_{i-1},s}(x)$.
\item[5:] Find
$$
s_i=
\argmax{s\in\cS_B\setminus\{s_1,\ldots,s_{i-1}\}}\lambda^\epsilon_{\min}\big(f^M_{s_1,\ldots,s_{i-1},s}(x)\big).
$$
\item[6:] If $i>k$, stop the algorithm. Otherwise, set $i=i+1$ and return to Step $4$.
\end{enumerate}

\Ensure
Subset $\cS_0=\{s_1,\ldots,s_k\}.$
\end{algorithmic}
\end{algorithm}

We have the following theorem for Algorithm $1$.

\begin{theorem}\label{th:alg}
Suppose that $0<\epsilon<\frac{1}{2k}$.
Algorithm $1$ can output a subset $\cS_0=\{s_1,\ldots,s_k\}$ such that
\begin{equation}
\label{algbound}
\big\|\big([A \ B_{\mathcal{S}_0}]\big)^{\dagger}\big\|_{\xi}^2 \leq \Gamma(m,n,k,r)\bigg(1+\frac{\|A^{\dagger}B\|^2_F}{m-n+r}\bigg)\left(1+2k\epsilon\right) \cdot\big\|([A \ \, B])^{\dagger}\big\|_{\xi}^2.
\end{equation}
The running time complexity is $O(k(m-\frac{k}{2})n^{\theta}\log (1/\epsilon))$,
where $\theta\in (2,2.373)$ is the matrix multiplication complexity exponent.
\end{theorem}

\begin{proof}
By Step $4$ in Algorithm $1$, we obtain that
$$
\lambda_{\min}\big(f^M_{s_1,\ldots,s_{k}}(x)\big)\geq
\lambda_{\min}\big(f^M_{s_1,\ldots,s_{k-1}}(x)\big)-\epsilon\geq
\cdots\geq \lambda_{\min}\big(f^M_{s_1}(x)\big)-(k-1)\epsilon\geq  \lambda_{\min}\big(f^M_{\emptyset}(x)\big)-k\epsilon.
$$
Then, using a similar argument for Theorem  \ref{fix-selection}, we can obtain the bound
(\ref{algbound}).
 We next establish the  running time complexity.

The main cost of Algorithm $1$ is  Steps $2$ and  $4$. In Step $2$, the time
complexity for  the computation of the SVD of  $[A \ \, B]$ is
$O\big((m+\ell)n^2\big)$. For Step $4$, at the $i$-th iteration, we claim the time
complexity for computing
$\lambda^\epsilon_{\min}\big(f^M_{s_1,\ldots,s_{i-1},s}(x)\big)$ over all
$s\in\cS_B\setminus \{s_1,\ldots,s_{i-1}\}$ is $O((m-i+1)n^{\theta}\log (
1/\epsilon))$. Therefore,  Algorithm $1$ can produce the subset $\mathcal{S}_0$  in
$O(k(m-\frac{k}{2})n^{\theta}\log (1/\epsilon))$ time.

Indeed, at the $i$-th iteration, the main cost of Step $4$ consists of
(\romannumeral1) the computations of $f^M_{s_1,\ldots,s_{i}}(x)$  and
(\romannumeral2) the computations of an $\epsilon$-approximation to the  smallest
root of $f^M_{s_1,\ldots,s_{i}}(x)$ for every $s_{i}\in\cS_B\setminus
\{s_1,\ldots,s_{i-1}\}$. First, for any $\{s_1,\ldots,s_j\}\subseteq\cS_B$, we can
compute the characteristic polynomial $p^M_{s_1,\ldots,s_j}(x)$  in $O(n^{\theta}\log
n)$  time, where $2<\theta<2.373$ is an admissible exponent of matrix multiplication
\cite{Gall2014,Keller1985}. From \eqref{alg-fm},  we know that the  time complexity
for the computation of $f^M_{s_1,\ldots,s_j}(x)$ is $O(n^{\theta}\log n)$ as its main
cost is to compute  $p^M_{s_1,\ldots,s_j}(x)$. Therefore, the running time for
computing $f^M_{s_1,\ldots,s_{i}}(x)$ over all
$s_{i}\in\cS_B\setminus\{s_1,\ldots,s_{i-1}\}$, which has $m-i+1$ choices, is
$O((m-i+1)n^{\theta}\log n)$.

Secondly, for any $\{s_1,\ldots,s_j\}\subseteq\cS_B$, we  can compute an
$\epsilon$-approximation to the smallest root of $f^M_{s_1,\ldots,s_{j}}(x)$  using
the standard technique of binary search with a Sturm sequence. This takes
$O(n^2\log(1/\epsilon))$ time per polynomial (see, e.g., \cite{Basu2007}). Noting
that $O((m-i+1)n^{\theta}\log n)+O((m-i+1)n^2\log
(1/\epsilon))=O((m-i+1)n^{\theta}\log (1/\epsilon))$, we obtain that the time
complexity  of Step $4$ is $O((m-i+1)n^{\theta}\log (1/\epsilon))$.
\end{proof}


\end{document}